\documentclass[runningheads,letterpaper]{llncs}

\usepackage[normalem]{ulem} 
\usepackage{amsmath}
\usepackage{amssymb}

\usepackage{wrapfig}
\usepackage{tikz}  %[border=5pt]
\usetikzlibrary{positioning,arrows}
\usetikzlibrary{decorations.pathmorphing,shapes}

\usepackage{algorithm}
\def\lf{\tiny}
\newcounter{linenumber}

\def\nnll{\refstepcounter{linenumber}\lf\thelinenumber}

\newtheorem{observation}[theorem]{Observation}

\newcommand{\safeagree}{\ensuremath{\mathit{SA}}}
\newcommand{\add}[1]{\textsc{add$(#1)$}}
\newcommand{\get}{\textsc{get()}}
\newcommand{\update}[1]{\textsc{update$(#1)$}}
\newcommand{\scan}{\textsc{scan()}}

\newcommand{\setagree}[1]{\textsc{setagree($#1$)}}
\newcommand{\propose}[1]{\textsc{propose(\ensuremath{#1})}}
\newcommand{\resolve}{\textsc{resolve()}}
\newcommand{\reduce}[2]{\ifthenelse{\equal{#1}{}}{%
    \textsc{reduce(\ensuremath{#2})}}{%
    \textsc{reduce\ensuremath{_#1}(\ensuremath{#2})}}}

\newcommand{\skel}{\ensuremath{\text{skel}}}
\newcommand{\bary}{\ensuremath{\text{bary}}}

\usepackage[margin=1in]{geometry}
\usepackage{newtxtext,newtxmath}
\makeatletter
    \renewcommand{\title}[1]{\gdef\@title{\bfseries\LARGE#1}}
\makeatother
\renewcommand{\abstractname}{\normalfont\large\bfseries Abstract}

\newcommand{\cI}{\ensuremath{\mathcal{I}}}

\newcommand{\cK}{\ensuremath{\mathcal{K}}}
\newcommand{\cL}{\ensuremath{\mathcal{L}}}

\newcommand{\cO}{\ensuremath{\mathcal{O}}}

\begin{document}

\mainmatter
\title{A characterization of colorless anonymous $t$-resilient task computability}
\titlerunning{A characterization of colorless anonymous $t$-resilient task computability} 

\author{Carole Delporte-Gallet\inst{1}
    \and Hugues Fauconnier\inst{1}
    \and Sergio Rajsbaum\inst{2}  
    \and Nayuta Yanagisawa\inst{3}  
}

\institute{
    IRIF-Universit\'{e} Paris-Diderot, France.\\ \email{\{cd,hf\}@irif.fr } \thanks{Supported by  LiDiCo.}
    \and Instituto de Matem\'aticas, UNAM,Mexico.\\ \email{rajsbaum@math.unam.mx} \thanks{Supported by UNAM-PAPIIT project IN104711.}
    \and Dept.\,of Mathematics, Graduate School of Science, Kyoto University, Kyoto.\\  \email{nayuta87@math.kyoto-u.ac.jp}
}

\authorrunning{C.\,Delporte, H.\,Fauconnier, S.\,Rajsbaum, and N.\,Yanagisawa} 
\toctitle{Lecture Notes in Computer Science}

\maketitle

\begin{abstract}
A \emph{task} is a distributed problem for $n$ processes, in which each process
starts with a private input value, communicates with other processes, 
and eventually decides an output value. 
A task is \emph{colorless} if each process can adopt the input or output value 
of another process. 
Colorless tasks are well studied in the non-anonymous shared-memory model
where each process has a distinct identifier that can be used to access 
a single-writer/multi-reader shared register. 
In the anonymous case, where processes have no identifiers and communicate through 
multi-writer/multi-reader registers, there is a recent topological characterization 
of the colorless tasks that are solvable when any number of asynchronous processes 
may crash. 
% where input values are taken from a finite set

In this paper we study the case where at most $t$ processes may crash, where $1\leq t<n$. 
We prove that a colorless task is $t$-resilient solvable non-anonymously if and only if 
it is $t$-resilient solvable anonymously.  This implies a complete characterization of 
colorless anonymous $t$-resilient asynchronous task computability.
\end{abstract}


%! TEX root = ./paper_v1.tex

\section{Introduction}
The central results in distributed task computability is the 
\emph{asynchronous computability theorem} (ACT)~\cite{herlihy1999topological}.
It characterizes the tasks that are solvable in asynchronous shared-memory systems 
where $n$ processes that may fail by crashing communicate by writing and reading 
shared registers.  
It is sometimes called the \emph{wait-free} characterization, because any number of 
processes may crash, and the processes are asynchronous (run at arbitrary speeds, 
independent from each other).
The characterization is of an algebraic topology nature. 
A \emph{task} is represented as a relation $\Delta$ between an input complex $\cI$
and an output complex $\cO$. 
Each simplex $\sigma$ in $\cI$ is a set that specifies the initial inputs to the processes 
in some execution.
The processes communicate with each other, and eventually decide output values that 
form a simplex $\tau$ in $\cO$. 
The computation is correct, if $\tau$ is in $\Delta(\sigma)$. 
The complex $\cI$ (resp. $\cO$) is \emph{chromatic} because each simplex specifies 
not only  input values, but also which process gets which input (resp.~output) value. 
Roughly, the ACT characterization states that the task is solvable if and only 
if there is a  simplicial map $\delta$ from a chromatic subdivision of $\cI$ to $\cO$ 
respecting $\Delta$. 
The map $\delta$ is also chromatic, because it sends an input vertex corresponding to 
process $p_i$ to an output vertex corresponding to the same process $p_i$.

The ACT is the basis to obtain a characterization of distributed task computability in
the case where at most $t$ asynchronous processes may crash, $1\leq t<n$. 
Also, it is the basis to study other distributed computing models, parametrized by 
the failure, timing and communication model, and even mobile robot models~\cite{robotsIPDPS17}. 
%There are even some initial results about long-lived tasks, where each process gets 
% a sequence of inputs and produces a sequence of outputs~\cite{crrNetys17}.
There are basically two ways of extending the results from the wait-free model to 
other models. 
One by directly generalizing the algorithmic and topological techniques, and another 
by reduction to other models using simulations (either algorithmic~\cite{BorowskyGLR01,gafni2009extended} 
or topological~\cite{HerlihyR12,saraph2016asynchronous}).
An overview of results in this area can be found in the book~\cite{herlihy2013distributed}.

The theory of distributed computing in~\cite{herlihy2013distributed} assumes that 
the processes, $p_0,\ldots,p_{n-1}$, communicate using single-writer/multi-reader (SWMR) 
registers, $R_0,\ldots,R_{N-1}$.
Thus, $p_i$ knows it is the $i$-th process, and it can write exclusively to $R_i$. 
However, often processes, while they know their ids, the number of possible ids $N$ 
is much bigger than the number of processes, $n$. 
In this situation, preallocating a register for each identifier would lead to a distributed 
algorithm with a very large space complexity, namely $N$ registers. 
Instead, it is shown in~\cite{DFGRbootstrapTCS15} that $n$ multi-writer/multi-reader (MWMR) 
registers are sufficient to solve any read-write wait-free solvable task. 
% The paper shows how the processes can non-blocking emulate a system of $n$ SWMR registers 
% on top of $n$ MWMR registers. Furthermore, in~\cite{DFGRbootstrapTCS15} it is shown
% that it is impossible to do such an emulation with $n-1$ MWMR registers.
 
However, in some distributed systems, processes are \emph{anonymous}; they  have no ids at all 
or they cannot make use  of their identifiers (e.g. due to privacy issues). 
Processes run identical programs, and the means by which processes access the shared memory 
are identical to all processes. 
A process cannot have a private register to which only this process may write, and hence
the shared memory consists only of MWMR registers.
This anonymous shared memory model of asynchronous distributed computing has been studied 
since early on~\cite{attiya2002computing,jayanti1991wakeup}, in the case where processes do not fail.
 
Only recently a characterization of the tasks that are wait-free solvable in the anonymous 
model has been given~\cite{yanagisawa2016wait-free}.
The characterization implies that the anonymity does not reduce the computational power of 
the asynchronous shared-memory model as long as \emph{colorless} tasks are concerned. 
Indeed, in an anonymous system, the task  specification must be colorless in the sense that 
it cannot refer to which process has which input or output value.
In consequence, the topological characterization is in terms of input and output complexes
which are not chromatic. 
Furthermore, the anonymous wait-free characterization matches exactly the eponymous
wait-free characterization for non-anonymous systems~\cite{herlihy2010topology,herlihy2017computingJ}.
%It is shown in~\cite{yanagisawa2016wait-free} that if a colorless task
%is solvable in a eponymous system, it is solvable in an anonymous system. In the journal version,
%weak set agreement is used and is shown that is solvable using only $n$ MWMR registers.

\paragraph{Results}
%In this paper we study the anonymous MRMR asynchronous model in the case where at most $t$ processes may crash.
Our main result is an extension of the wait-free characterization of~\cite{yanagisawa2016wait-free},
to the case where at most $t$ processes may crash, $1\leq t<n$.
We prove that a  colorless task is $t$-resilient solvable anonymously if and only if it is 
$t$-resilient solvable non-anonymously. 

Our main result is to show that if a colorless task is $t$-resilient solvable non-anonymously, 
then it is solvable anonymously using only $n$ MWMR registers.
The result is obtained through a series of reductions depicted in the figure below. 
We hope they provide useful basic tools to study further anonymous fault-tolerant computation.
First, we design an anonymous non-blocking implementation of an atomic \emph{weak set object} 
with $n$ registers.
The construction is based on the non-blocking atomic snapshot of~\cite{Ellen2008,guerraoui2007anonymous}.
Then we build a wait-free implementation, of a safe agreement object for an arbitrary value set $V$.
Our implementation is a generalization of the anonymous consensus algorithm proposed 
in~\cite{attiya2002computing}.
We then describe two ways of deriving the $t$-resilient anonymous solvability characterization.
One way is through a novel anonymous implementation of the BG-simulation~\cite{BorowskyGLR01},
which we use to simulate a non-anonymous system by an anonymous system, both $t$-resilient.
The other way is to use the safe-agreement object to solve $k$-set agreement, and then do
the topological style of analysis~\cite{HerlihyR12} and~\cite{yanagisawa2016wait-free}.
\tikzset{
    %Define standard arrow tip
    >=stealth',
    %Define style for boxes
    punkt/.style={
           rectangle,
           rounded corners,
           draw=black, very thick,
           text width=6.5em,
           minimum height=2em,
           text centered},
    % Define arrow style
    pil/.style={
           ->,
           thick,
           shorten <=2pt,
           shorten >=2pt,}
}
\begin{wrapfigure}[21]{i}{-0.1\textwidth}%   \centering
%                          ^^ This dictates the number of text rows the wrapfigure will occupy.                   
 \centering
 %\vspace{-40mm}
\begin{tikzpicture}
  \node [punkt] (a) {MWMR atomic registers};
  \node [punkt] (b) [below=of a] {atomic weak set};
  \node [punkt] (c) [below=of b] {safe agreement};
    \node (dummy) [below=of c] {$\,$};
    \node [punkt] (d) [right=of dummy] {$k$-set agreement};
    \node [punkt] (e) [below=2 of c] {characterization of  $t$-resilient solvability};
      \draw[pil,->] (a) to node [right] {(non-blocking)} (b);
  \draw[pil,->] (b) to node {} (c);
  \draw[pil,->] (c) to node {} (d);
    \draw[pil,->,  right] (d) to node {} (e);%\cite{herlihy2010topology}+\cite{yanagisawa2016wait-free}
      \draw[pil,->,bend right=45] (c) to node [left] {\parbox{2.5cm}{anonymous  BG~simulation}} (e);
\end{tikzpicture}
%\vspace{-4mm}
\end{wrapfigure}
%\vspace{3mm}

\paragraph{Related work}
Colorless tasks include many tasks such as consensus~\cite{fischer1985impossibility}, 
set agreement~\cite{Chaudhuri1993more}, and loop agreement~\cite{herlihy2003classification}, 
and have been widely studied (in the non-anonymous case). 
The first part of the book~\cite{herlihy2013distributed} is devoted to colorless tasks.
% [22,23,25,31].
Colorless tasks were identified in~\cite{BorowskyGLR01} as the ones for which 
the BG-simulation works, and in~\cite{herlihy1997decidability} for the purpose of showing 
that they are undecidable in most distributed computing models. 
Not all tasks of interest are colorless though, and general tasks can be much harder to study, 
e.g.~\cite{CastanedaRenam:2011,gsb2016}.

A characterization of the colorless tasks that are solvable in the presence of processes 
that can crash in a dependent way is provided in~\cite{herlihy2010topology}, and a 
characterization when several processes can run solo is provided in~\cite{herlihy2017computingJ}. 
Both encompass the wait-free colorless task solvability characterization, and the former 
encompasses the $t$-resilient characterization that we use in this paper.

A certain kind of anonymity has been considered in~\cite{herlihy1999topological} to establish 
the anonymous computability theorem. 
However, they allow SWMR registers while we assume a fully anonymous model with only MWMR registers.

Anonymous  distributed computing remains an active research area since the shared-memory seminal 
papers~\cite{attiya2002computing,jayanti1991wakeup} and the  message-passing paper~\cite{angluin1980local}.
For some recent papers and references herein see, e.g.~\cite{CapdevielleAnonAg2017,GelashviliAnonCons2015}.

Closer to our paper is~\cite{guerraoui2007anonymous} where the anonymous asynchronous MWMR fault-tolerant 
shared-memory model is considered.
Our weak atomic set object provides an enhanced atomic implementation of the weak set object supporting 
non-atomic operations  presented in~\cite{delporte-gallet2009two}.
A set object that also supports a remove operation, but satisfies a weaker consistency condition, 
called per-element sequential consistency is presented in~\cite{baldoni2010value-based}. 

\paragraph{Organization}
In Section~\ref{sec:prelim} we briefly recall some of the notions used in this paper, about
the model of computation and the topology tools, both of which are standard.
In Section~\ref{sec:aws} we present the  anonymous  implementation of an atomic weak set object 
from MWMR registers.
In Section~\ref{sec:safeAg} we present the safe agreement implementation.
In Section~\ref{sec:charact} we derive our anonymous characterization of the $t$-resilient solvability 
of colorless tasks.

% !TEX root =  ./paper.tex

\section{Preliminaries}
\label{sec:prelim}
We assume a standard \emph{anonymous asynchronous shared-memory model}%
~\cite{guerraoui2007anonymous} consisting of $n$ sequential processes 
that have no identifiers and execute an identical code.
We assume tha at most $t$ of the processes may fail by crashing, where $1\leq t<n$.
Processes are asynchronous, i.e., they run at arbitrary speeds, independent from each other.
The processes communicate via multi-writer/multi-reader (MWMR) registers.
Let $R[0\ldots m-1]$ denote an array of $m$ registers. 
The read operation, denoted by $\textsc{read}$, returns the state of  $R[i]$. The write operation, 
denoted by $\textsc{write}(i,v)$, changes the state of $R[i]$ to $v$ and 
returns $ack$. 
The registers are assumed to be atomic (linearizable)~\cite{herlihy1990linearizability}. 
We assume that the registers are initialized to some default value.
We sometimes refer to the processes by unique names $p_0$, $\ldots$, $p_{n-1}$
for the convenience of exposition, but processes themselves have no means to
access these names.
Let us write $\Pi=\{p_0,\ldots,p_{n-1}\}$.

A \emph{complex} $\cK$ on a finite set $V(\cK)$ of vertices is a family of nonempty subsets 
of $V(\cK)$, called \emph{simplices}, such that $\{v\}\in\cK$ for every $v\in V(\cK)$. 
Also, $\cK$ is closed under containment, meaning that if $s\in\cK$ then $s' \in\cK$, 
for every $s' \subseteq s$.
A subset of a simplex $s$ is called a \emph{face} of $s$. 
The \emph{dimension} of $s$ is $\#s-1$, where $\#V$ denotes the cardinal of a set $V$.
% When a complex $\cK$ is a subset of a complex $\cL$, we say that $\cK$ is a 
% \emph{subcomplex} of $\cL$.
A map $\phi:V(\cK)\to V(\cL)$, where $\cK$ and $\cL$ are complexes,
is said to be a \emph{simplicial map}, if $\phi(\sigma)\in\cL$ for every $\sigma\in\cK$.
We can associate any complex $\cK$ to the corresponding topological space
$|\cK|\subseteq\mathbb{R}^d$, for a sufficiently large integer $d$,
by embeding vertices of each simplex into $\mathbb{R}^d$
in an afinaly independent way and taking convex hull of these vertices.

The \emph{barycentric subdivision} of a complex $\cK$, denoted by $\bary\cK$,
is the complex such that $V(\cK)=\cK$ and 
a set $\{s_0,\ldots,s_i\}\subseteq \cK$ is a simplex of $\bary\cK$
if and only if $s_0,\ldots,s_i$ are totally ordered by containment.
The \emph{$b$-iterated barycentric subdivision}  of a complex $\cK$,
denoted by $\bary^b \cK$, is defined by $\bary(\bary^{b-1}\cK)$, where $\cK^0=\cK$.
The \emph{$k$-skeleton} of a complex $\cK$, denoted by $\skel^k \cK$,
is the complex whose simplices are the simpleces of $\cK$
of dimension less than or equal to $k$.

Let $I$ and $O$ be complexes. A \emph{carrier map} from $\cI$ to $\cO$ 
is a mapping $\Delta:\cI\to 2^\cO$ such that, for each $s\in\cI$, 
$\Delta(s)$ is a subcomplex of $\cO$ and $s'\subseteq s$ implies $\Delta(s')\subseteq \Delta(s)$.  
% Let $\Delta\subseteq\cI\times\cO$ be a carrier map.
If a continuous map $f:|\cI|\to|\cO|$ satisfies $f(|\sigma|)\subseteq|\Delta(\sigma)|$
for all $\sigma\in\cI$, we say that $f$ is carried by $\Delta$.
If a simplicial map $\delta:\bary^b\cI\to\cO$ satisfies 
$\delta(\bary^b \sigma)\subseteq \Delta(\sigma)$ for all $\sigma\in\cI$,
we say that $\delta$ is \emph{carried by} $\Delta$.
As an immediate consequence of Lemma 3.7.8. of \cite{herlihy2013distributed},
the following lemma holds.
\begin{lemma}\label{thm:approximation}
    If $\Delta:\cI\to 2^\cO$ is a carrier map and $f:|\cI|\to|\cO|$ 
    is a continuous map carried by $\Delta$, then there is a non-negative integer $b$ 
    and a simplicial map $\delta:\bary^b\cI\to\cO$ carried by $\Delta$.
\end{lemma}

A \emph{colorless task} is a triple $T = (\cI,\cO,\Delta)$, where $\cI$ and $\cO$ are simplicial complexes
and $\Delta$ is a carrier map. A colorless task $T$ is solvable, if for each input simplex $s\in\cI$, 
whenever each process $p_i$ starts with input value $v_i\in s$ (different processes may start with the same value), 
eventually it decides an output value $v'_i$, such that the set of output values form a simplex $s'\in\Delta(s)$.  
The colorless tasks that are fundamental to the present paper are the \emph{$b$-iterated barycentric agreement} 
and the \emph{$k$-set agreement}.
The $b$-iterated barycentric agreement is a colorless task $T=(\cI,\bary^b \cI,\bary^b)$, where we write by $\bary^b$ 
the carrier map that maps $s\in \cI$ to $\bary^b s$ for an abuse of notation.
The $k$-set agreement is a colorless task $T_k=(\cI,\skel^k \cI,\skel^k)$, where $\skel^k$ denotes the carrier map 
that maps a simplex $s\in \cI$ to the subcomplex $\skel^k \cI$.

%!TEX root = paper.tex

\section{Atomic weak set }\label{sec:aws}
Here, we present an anonymous implementation of an atomic weak set object on an arbitrary value set $V$.

\subsection{Specification and Algorithm}

\newcommand{\weakset}{\ensuremath{\mathit{SET}}}
\newcommand{\ack}{\ensuremath{\mathit{ACK}}}

An \emph{atomic weak set} object, denoted by \weakset{}, is an atomic object used for storing values. 
The object supports only two operations, add and get, and has no remove operation, which is why 
it is called ``weak."
The add operation, denoted by \add{v}, takes an argument $v\in V$ and returns \ack{}.
The get operation, denoted by \get{}, takes no argument and returns the set of values 
that have  appeared as arguments in all the \add{} operations preceding the \get{} operation. 
We assume that \weakset{} initially holds no values, i.e., it is $\emptyset$.

We assume that a non-blocking atomic \emph{snapshot} object is available.
% constructed on top of the array of $n$ MWMR registers. 
An implementation in  an anonymous setting with $n$ registers is described 
in~\cite{Ellen2008,guerraoui2007anonymous}.
The snapshot object exports two operations, \update{} and \scan{}.
Informally, a $scan()$  returns an array of $n$ values, which are contained in the array
of $n$ MWMR registers at some point in time between the invocation and the response of the 
$scan()$ operation.

We propose an anonymous non-blocking implementation of the atomic weak set object on $n$ MWMR registers.
The pseudocode of the implementation appears in Fig.~\ref{fig:aws}. 
If $S$ is an array of $n$ cells, we denote $vals(S)=\{S[i]| i \in \{1,...,n\}\}$.
The idea of the algorithm is as follows.  
To execute an \add{v} operation, the algorithm repeatedly tries to store the value $v$ 
in each one of the $n$ components of the snapshot object, using an $update$ operation (line~\ref{li:update1})
until it detects that $v$ appears in all the components. 
In each iteration, the algorithm deposits in the snapshot object not only $v$, but $View$, 
containing all the values known to be in the set so far, including $v$ itself. 
Once $v$ is detected to be in all components of the snapshot object, the \add{v} terminates.
The \get{} operation is similar, except that now the $View$ of the process has to appear 
in all the components of the snapshot for the operation to terminate.
Intuitively, once a value $v$ (or a set of values) appears in all $n$ components of the snapshot
object, it cannot be overwritten and go unnoticed by other processes, because the other processes
can be covering (about to write) at most $n-1$ components.
\begin{figure}[h]
\hrule \vspace{1mm} {\small
\setcounter{linenumber}{0}
\begin{tabbing}
bbb\=bbb\=bbb\=bbb\=bbb\=bbb\=bbb\=bbb\=bbb\=bbb\=bbb\=bbb\=  \kill

Shared variable : \\
\>\> \texttt{array of} $n$ \texttt{MWMR-register} :  $R[0\ldots n-1]$\\[2mm]

\textsc{Code for a process } \\ 
Local variable:\\
         \>\>\texttt{array of} $n$ \texttt{set of Values} $Snap[0\ldots n-1]$\\
	\>\>\texttt{set of Values} $View$ init $\emptyset$\\
	\>\>\texttt{integer} $next$\\[1mm]
Macro:\\
$vals(Snap)=\cup Snap[i]$\\[1mm]
%$vals(Snap)$ is the union of the values in all cells of $Snap$\\[1mm]

\textsc{add}($v$):\\

\nnll\>$next=0$ \\

\nnll\label{li:scan1}\>$Snap= R.\scan{}$\\
\nnll\>$View= View \cup vals(Snap) \cup  \{v\}$\\
\nnll \label{li:while-add}\> \textbf{while} $(\# \{r |v $ in $ Snap[r]\} <  n)$\\

\nnll\label{li:update1}\>\>$R.\update{next,View}$\\
\nnll\>\>$next=(next + 1) \bmod{ n }$\\

\nnll\label{li:scan2}\>\>$Snap=R.\scan{}$\\
\nnll\>\> $View= vals(Snap) \cup View$\\
\nnll\>\textbf{return} $ACK$ \\

\\
\textsc{get}: \\
\nnll\>$next=0$ \\
\nnll\label{li:scan3}\>$Snap= R.\scan{}$\\
\nnll\>$View= vals(Snap) \cup View$\\
\nnll \label{li:while-get}\> \textbf{while} $(\# \{r |View= Snap[r]\} <   n)$\\

\nnll\label{li:update2}\>\>$R.\update{next,View}$\\
\nnll\>\>$next=(next+ 1) \bmod{ n }$\\

\nnll\label{li:scan4}\>\>$Snap= R.\scan{}$\\
\nnll\>\> $View= vals(Snap) \cup View$\\

\nnll\>\textbf{return} $View$ \\
\end{tabbing}
\vspace{-6mm}
\hrule}
\caption{non-blocking implementation of atomic weak set for $n$ processes.} 
\label{fig:aws}
\end{figure}

\subsection{Correctness of weak set object implementation}

\subsubsection{Safety}

%%% MODEL
Given an operation $op$,  $invoc (op)$ denotes its invocation and
$resp(op)$ its response.

% History $H$; sequential history $H_{seq}$

Let $H$ be a history of the algorithm as defined in~\cite{herlihy1990linearizability}. 
$H_{seq}$ denotes the sequential history in which each operation of $H$ appears as 
if it has been executed at a single point (the linearization) of the time line.
We have to define linearization points and prove that:
\begin{itemize}
    \item the linearization point of each operation \get{} and \add{} appear between 
        the beginning and the end of this operation;
    \item the sequential history that we get with these points respect the sequential 
        specification of the weak set. 
\end{itemize}
Most of the details of the proofs are in the appendix.

Consider a history $H$, let $v$ be  a value or a set of values, define time $\tau_{v}$
as the first time, if any, that $v$ belongs to all registers of $R$. 
When there is no such time,  $\tau_v$ is $\bot$.

\begin{lemma}\label{lemma:terminaison}
    If the operation \add{v} terminates, then before the end of this operation $v$ belongs 
    to all registers of $R$. 
    If the operation \get{} terminates  and returns $V$, then before the end of this 
    operation $V$ belongs to all registers of $R$. 
\end{lemma}

By Lemma~\ref{lemma:terminaison}, $\tau_v$ is not $\bot$ for each operation
\add{v} that terminates and $\tau_V$ is also not $\bot$ for each operation \get{} that
terminates  and returns $V$.

\paragraph {\it Linearization points for operations \add{} and \get{}}:
\begin{itemize} 
    \item $op=\textsc{add}(v)$:
        If $\tau_v\neq \bot$, the linearization point $\tau_{op}$ of an
        operation $op=\textsc{add}(v)$ is $max\{ \tau_{v}, 
        invoc(op) \}$.
        If $\tau_v=\bot$, the operation $op$ does not terminate and   
        is not linearized.
    \item
        $op=\textsc{get}()$:
        The linearization point $\tau_{op}$ of an operation
        $op=\textsc{get}()$ that returns $V$ is $max\{ \tau_{V}, 
        invoc(op) \}$. A $\textsc{get}()$ operation that does not terminate is
        not linearized. 
\end{itemize}
Directly from the definition and Lemma\ref{lemma:terminaison}, the linearization point 
$\tau_{op}$ appears between the invocation and the response of $op$:
\begin{lemma}
    Let $op$ be an operation of $H$. If $\tau_{op}$ is defined, then $
    \tau_{op}$ belongs to $[invoc(op), resp(op)]$. 
    If $\tau_{op}$ is undefined, then $op$ does not terminate and is not linearized.
\end{lemma}

In the following, we consider the next snapshot operation on $R$ of each process.
This next snapshot operation is either a \scan{} or an \update{} or there is no next 
snapshot operation (when the process is about to satisfy the termination loop condition: 
$\#\{r|v \in Snap[r] \}\geq n$ for an \add{} and $\#\{r|View= Snap[r] \}\geq n$ for an \update{}).

\noindent
Consider any time $\tau$ and define $r_v(\tau)$, $w_v(\tau)$, $c_v(\tau)$
and $\alpha_v(\tau)$ as follows:
\begin{itemize}
    \item $r_v(\tau)$ is the number of processes for which, after time
        $\tau$, the next snapshot operation is a \scan{};
    \item $w_v(\tau)$ is the number of processes such that (1) $v \in View$ at time $\tau$ 
        and for which after time $\tau$ the next snapshot operation is an \update{},
        or (2) there is no next snapshot operation for that process (the process has finished
        -or is going to finish- its main loop or it takes no more steps);
    \item $c_v(\tau)$ is the number of registers that contains $v$ at time $\tau$;
    \item $\alpha_v(\tau)$ is defined by $\alpha_v(\tau)=  r_v(\tau)+w_v(\tau) +c_v(\tau)$.
\end{itemize}

\noindent
As soon as $\alpha_v(\tau) >n$, $\alpha_v(\tau) $ is not decreasing:
\begin{lemma}\label{lemma:increasingstep}
    Assume $\alpha_v(\tau) >n$ and the next step in $H$, is made at time $\tau'\geq \tau$,
    then we have  $\alpha_v(\tau)\leq \alpha_v(\tau')$.
\end{lemma}
\noindent
And by an easy induction on the steps of $H$ we get:
\begin{lemma}\label{lemma:increasing0}
    If $\alpha_v(\tau) > n$  then for all $\tau'$, such that $\tau \leq \tau' $, 
    $\alpha_v(\tau) \leq \alpha_v(\tau')$.
\end{lemma}
When $\tau_v$ is defined, we can verify that $\alpha_v(\tau) > n$ then:
\begin{lemma}\label{lemma:increasing}
    If $\tau_v\neq \bot$ then   for all $\tau$,  such that $\tau_v\leq
    \tau $, $n<\alpha_v(\tau_v)\leq \alpha_v(\tau)$.
\end{lemma}

\noindent
From the previous Lemmas we may deduce:
\begin{lemma}\label{lemma:safe}
$H_{seq}$ satisfies the sequential specification of the weak set.
\end{lemma}

\subsubsection{Liveness}

We prove that the algorithm is non-blocking, namely, if processes perform operations forever, 
an infinite number of  operations terminates.

By contradiction, assume that there is only a finite number of operations \get{} and \add{} 
and some operations made by correct processes do not terminate. 

Operations \add{} or \get{} may not terminate because the termination conditions of the while
loop are not satisfied (Lines~\ref{li:while-add} or~\ref{li:while-get}): 
for an \add{v} operation, in each \scan{} made by the process, $v$ is not in at least one of 
the registers of $R$, and for a \get{} operation, in each $snap$, all the registers are not
equal to the view of the process.

There is a time $\tau_0$ after which there is no new process crash and all 
processes that terminate \get{} or \add{} operations in the run have already 
terminated. 
Consider the set $N$ of processes alive after time $\tau_0$ that do not terminate 
operations in the run. 
Note that after time $\tau_0$ only processes in $N$ take steps and as no process 
in $N$ may crash each process in $N$ makes an infinite number of steps. 

%%%%
%From observation~\ref{obs:integrity} 
We notice that all values in  variables $View$ have been proposed by some \get{}.
%%%%
If there is a finite number of operations, then all variables $View$ are subsets 
of a finite set of values.
Moreover,
%%%%%%from observation~\ref{obs:increasing-view} 
considering the inclusion $\subseteq$, the views of each process are increasing, 
then there is a time $\tau_1>\tau_0$ after which the view of each process $p$ 
in $N$ converges to a \emph{stable view} $SView_p$: 
after time $\tau_1$ forever  the view of $p$ is $SView_p$.
In the following $SV$ denotes $\{SView_p| p\in N\}$ the set of all stable
views for processes in $N$. Observe that:
\begin{observation}\label{obs:vinView}
    If  $p$ does not terminate an $\textsc{add}(v)$ then $v\in SView_p$.
\end{observation}

After time $\tau_1$ processes only update $R$ with their stables views $SView$, 
and as each process in $N$ updates infinitely often each register of $R$ with
its $SView$ there is a time $\tau_2$ after which all registers in $R$ contain 
only stable views of processes in $N$:
\begin{observation}\label{obs:st2}
    After time $\tau_2$ for all $i$, $R[i]\in SV$.
\end{observation}

Among stable views consider any minimal view $SView_0$ for inclusion,
i.e. for all $S\in SV$, $S\subseteq SView_0$ implies $SView_0=S$.

Consider any process $p\in N$ having the $SView_0$ as stable view, eventually 
$p$ makes a scan of the memory $R$ (Line~\ref{li:scan1} or~\ref{li:scan2} 
for \add{},  Line~\ref{li:scan3} or~\ref{li:scan4} for \get{}). 
Let $Snap$ be the array returned by the $scan$.  $Snap$ is the value of 
the array of registers $R$ at some time after $\tau_2$. 
Then $p$ adds $\bigcup_{1\leq i\leq n}  Snap[i]$ to its view $SView_0$.
The $SView_0$ being stable we have $\bigcup_{1\leq i\leq n} Snap[i]\subseteq SView_0$ 
and then for all $i$, $Snap[i] \subseteq   SView_0$.
But by Observation~\ref{obs:st2}, $R[i]=Snap[i]$ is a stable view $S \in SV$. 
Then by the minimality of $SView_0$,  for all $i$ we have $Snap[i]=SView_0$.
Then consider the two following cases:
\begin{itemize}
    \item
        if $p$ is performing an $\textsc{add}(v)$, as $v \in SView_p=SView_0$, 
        for all $i$, $v\in Snap[i]$ and the loop condition $(\#\{r |v $ in
        $ Snap[r]\} <  n)$ (Line~\ref{li:while-add}) is false and $p$
        terminates $\textsc{add}(v)$-- A contradiction
    \item
        if $p$ is performing an $\textsc{get}()$, then the loop continuation condition 
        $(\# \{r |View= Snap[r]\} <   n)$ is false and $p$ terminates
        operation $\textsc{get}()$ --- A contradiction
\end{itemize}
We deduce that there is no minimal stable view proving that $SV=\emptyset$
and also $N=\emptyset$.

%%% Local Variables:
%%% mode: latex
%%% TeX-master: "paper"
%%% End:
 
% !TEX root =  ./paper.tex
%\section{Preliminaries}

\section{Safe Agreement Object}\label{sec:safeAg}

An \emph{safe agreement object}~\cite{BorowskyGLR01} provides two operations, 
\propose{} and \resolve{}.
A \propose{} operation takes an argument $v\in V$ and returns \textit{ACK}.
A \resolve{} operation takes no argument and returns $u\in V$ or $\perp$.
The safe agreement object is \emph{one-shot}, i.e., each process may perform 
at most one \propose{} operation on each object, while an arbitrary number of 
\resolve{} operation can be invoked on a single object.
The object satisfies the following five conditions:
\begin{description}\setlength{\parskip}{0pt} 
\item[Validity] 
    Any non-$\perp$ value returned by a \resolve{} operation is an argument of 
    some \propose{} operation;
\item[Agreement] 
    If two \resolve{} operations return non-$\perp$ values $v$ and $v'$, 
    then $v=v'$;  
\item[Termination] 
    Every operation performed by a non-faulty process eventually terminates;
\item[Nontriviality] 
    If more than one \propose{} operations are performed and no process fails 
    while performing these \propose{} operations, every \resolve{} operation 
    started after some time instance returns non-$\perp$ value.
\end{description}
The above specification of the safe agreement object is based on
\cite{attiya2006adapting,bouzid2016anonymity}.

We propose an anonymous wait-free implementation, presented in Fig.~\ref{alg:safe_agreement}, 
of the safe agreement object for an arbitrary value set $V$.
The implementation makes use of $n$-array of weak set objects $\weakset[0\ldots n-1]$.
Each process firstly sets its input value to a local variable $view$. 
The process repeats the following procedure for $i=0,\ldots,n-1$:
it adds $view$ to the $\weakset[i]$; 
if $\weakset[i]$ holds a set of cardinality more than 2 and $view$ is the minimum 
value of the set, it waits until it gets a non-empty set from the $\weakset[n-1]$; 
otherwise, it sets the minimum value of the set to $view$. 
When a process gets a non-empty set from the $\weakset[n-1]$, it returns the minimum 
value in the set.

Our implementation is a generalization of the anonymous consensus algorithm 
proposed by Attiya et al.~\cite{attiya2002computing}.
Bouzid and Corentin~\cite{bouzid2016anonymity} have proposed an anonymous 
implementation of the safe agreement object for the case of $V=\{0,1\}$, 
where their implementation is also based on~\cite{attiya2002computing}.
However, their implementation is not immediately extended to the case of 
the infinite value set.
\begin{figure}[h]
\hrule \vspace{1mm} {%\small
\setcounter{linenumber}{0}
\begin{tabbing}
bbb\=bbb\=bbb\=bbb\=bbb\=bbb\=bbb\=bbb\=bbb\=bbb\=bbb\=bbb\=  \kill

Shared variable : \\
    \>\texttt{array of atomic weak set objects} : $\weakset[0\ldots n-1]$\\[2mm]

\textsc{Code for a process} \\[2mm] 
Local variable:\\
    \>\texttt{Value} $\mathit{view}$ init $\perp$\\
    \>\texttt{Integer} $i$ init 0\\
    \>\texttt{set of Values} $\mathit{Snap}$ init $\emptyset$\\[2mm]

\textbf{operation} \propose{v}:\\
\nnll\> $\mathit{view}=v$\label{line:input}\\
\nnll\> \textbf{for} $i=0,\ldots,n-1$ \textbf{do}\\
\nnll\>\>    $\weakset[i].\add{\mathit{view}}$\\
\nnll\label{line:get}
    \>\>    $\mathit{Snap}=\weakset[i].\get{}$\\
\nnll\label{line:if}
\>\>   \textbf{if} $\#\mathit{Snap}\ge2$ \&\& $\mathit{view}==\min(\mathit{Snap})$ \textbf{then}\\
% \nnll\>\>\>     \textbf{return} false\\
\nnll\>\>\>     \textbf{return} \textit{ACK}\\
\nnll\>\>   \textbf{else}\\
\nnll\label{line:update}
    \>\>\>     $\mathit{view}=\min(\mathit{Snap})$\\
% \nnll\>\>   \textbf{endif}\\
% \nnll\> \textbf{endfor} \\[2mm]
\nnll\>     \textbf{return} \textit{ACK}\\[2mm]
% \nnll\> \textbf{return} true\\[2mm]

\textbf{operation} \resolve{}:\\
\nnll\>   $\mathit{Snap}=\weakset[n-1].\get{}$\\
\nnll\> \textbf{if} $\mathit{Snap}\neq \emptyset$ \textbf{then}\\
\nnll\>\>   \textbf{return} $\min(\mathit{Snap})$\\
\nnll\> \textbf{else}\\
\nnll\>\>   \textbf{return} $\perp$\\
% \nnll\> \textbf{endif}\\
\end{tabbing}
\vspace{-6mm}
\hrule}
\caption{Anonymous implementation of safe agreement object}
\label{alg:safe_agreement}
\end{figure}

We now prove the correctness of the algorithm of Fig.~\ref{alg:safe_agreement}.
Recall that, although we refer to the processes by unique names $p_0$, $\ldots$, $p_{n-1}$, 
processes themselves have no means to access these names.
These set of names is denoted by $\Pi=\{p_0,\ldots,p_{n-1}\}$.
%The basic structure of some part of the proof is similar to that of
%the proof of Theorem 4.1 in~\cite{attiya2002computing}.

\begin{lemma}\label{thm:inclusion}
    Fix an execution of the algorithm of Fig.~\ref{alg:safe_agreement}.
    Let $V_i$ be the set of all the values that are added to $\weakset[i]$
    in the execution. 
    Then, $V_{i}\supseteq V_{i+1}$ holds for all $i=0,\ldots,n-2$.
\end{lemma}

\begin{proof}
    Every value added to $\weakset[i+1]$ is a value held in $\weakset[i]$
    by Lines~\ref{line:get} and \ref{line:update} of the algorithm.
    Thus, the lemma follows.
\end{proof}

\begin{lemma}[Validity]\label{thm:validity}
    The algorithm of Fig.~\ref{alg:safe_agreement} satisfies the validity condition.
\end{lemma}

\begin{proof}
    Fix an execution and define $V_i$ in the same way as Lemma~\ref{thm:inclusion}.
    By Line~1, every value in $V_0$ is an argument of some \propose{} operation.
    By Lemma~\ref{thm:inclusion}, $V_0\supseteq\cdots\supseteq V_{n-1}$
    and thus every value in $V_{n-1}$ is also an argument of some \propose{} operation.
    This completes the proof because every non-$\perp$ value returned by a \resolve{} 
    operation is a value held in the object $\weakset[n-1]$. 
\end{proof}

\begin{lemma}\label{thm:decrease}
    Fix an execution of the algorithm of Fig~\ref{alg:safe_agreement},
    in which more than one \propose{} operations are performed.
    Let $V_i$ be the set of the all values that are added to $\weakset[i]$
    in the execution. Let 
%    We write $U_i=V_i\setminus\{\min V_i\}$ and
    \[
    \Pi_i=\{p\in\Pi\mid\text{$p$ performs \propose{} and adds $v\in V_i\setminus\{\min V_i\}$ to $\weakset[i]$}\}.
    \]
    Then, $\Pi_{i}\supseteq\Pi_{i+1}$ holds for all $i=0,\ldots,n-2$.
\end{lemma}

\begin{proof}
    % Note: We agreed to prove this lemma by contradiction. 
    %       However, I think that the proof presented below is more understandable
    %       than one by contradiction.
    If $p\not\in\Pi_i$, there are two cases: 
    the process $p$ does not execute the $i$-th iteration of the for loop;
    the process $p$ executes the $i$-th iteration and adds $\min V_i$ to $\weakset[i]$.
    In the former case, the process does not execute the $(i+1)$-th iteration
    and thus $p\not\in\Pi_{i+1}$.
    In the latter case, the process adds $\min V_{i}$ to $\weakset[i+1]$ 
    (if it does not fail)
    and thus $p\not\in \Pi_{i+1}$ because $\min V_i \le \min V_{i+1}$
    by Lemma~\ref{thm:inclusion}.
    In either case, $p\not\in\Pi_i$ implies $p\not\in\Pi_{i+1}$.
    We obtain $\Pi_{i}\supseteq\Pi_{i+1}$ by taking contrapositive.
\end{proof}

\begin{lemma}[Agreement]\label{thm:agreement}
    The algorithm of Fig.~\ref{alg:safe_agreement} satisfies the agreement condition.
\end{lemma}

\begin{proof}
    If no \propose{} operation are performed in an execution,
    it is clear to see that every \resolve{} operation returns $\perp$.
    Thus, the lemma trivially holds for this case.

    Fix an execution, in which more than one \propose{} operations are performed
    and define $V_i$ and $\Pi_i$ in the same way as Lemma~\ref{thm:decrease}. 
    Note that $\#\Pi_{0}\le n-1$ by definition.
    It is sufficient to prove that $\#V_{n-1}<2$ for any execution,
    because every non-$\perp$ value returned by a \resolve{} operation 
    is a value held in the object $\weakset[n-1]$. 
    We prove this by contradiction.
    
    Assume that $\#V_{n-1}\ge 2$ holds.
    This implies $\#V_i\ge2$ for all $i=0,\ldots,n-1$ by Lemma~\ref{thm:inclusion}.
    Let us define
    \[
        \Sigma_i=\{p\in\Pi\mid\text{$p$ performs \propose{} and adds $\min V_i$ to $\weakset[i]$}\}.
    \]
    Note that $\Pi_i\cap\Sigma_i=\emptyset$.

    We now prove that $\Pi_{i}\supsetneq\Pi_{i+1}$ for all $i=0,\ldots,n-2$.
    We consider the following two cases:
    
    \emph{Case 1}:
    Suppose that some process $p\in\Pi_i$ sees $\min V_i$ at Line~\ref{line:get} 
    in the $i$-th iteration.
    Then, the process $p$ assigns $\min V_i$ to its $\mathit{view}$ and adds $\min V_i$ 
    to $\weakset[i+1]$ in the $(i+1)$-th iteration.
    This leads $p\in\Sigma_{i+1}$ and thus $p\not\in\Pi_{i+1}$ because $\min V_i\le\min V_{i+1}$ by 
    Lemma~\ref{thm:inclusion}.
    
    \emph{Case 2}:
    Suppose that no process in $\Pi_i$ sees $\min V_i$ at Line~\ref{line:get}
    in the $i$-th iteration.
    Then, the processes in $\Sigma_i$ perform \add{} operations on $\weakset[i]$
    after the processes in $\Pi_i$ perform \get{} operations on $\weakset[i]$.
    This leads that all the processes in $\Sigma_i$ sees $V_i$
    in the $i$-th iteration and jump to the while loop
    and thus $\min V_i<\min V_{i+1}$ by Lemma~\ref{thm:inclusion}. 
    Thus, at least one process in $\Pi_{i}$ adds $\min V_{i+1}$ to $\weakset[i+1]$
    and is not in $\Pi_{i+1}$.
    
    In either case, $\Pi_{i}\supsetneq\Pi_{i+1}$ holds by Lemma~\ref{thm:decrease}.
    $\#\Pi_{0}\le n-1$ and $\Pi_{0}\supsetneq\cdots\supsetneq\Pi_{n-1}$
    imply $\Pi_{n-1}=\emptyset$.
    Thus, only $\min V_{n-1}$ is added to the object $WS_{n-1}$ and $\#V_{n-1}=1$.
    This leads contradiction.
\end{proof}

 It is clear to see that the algorithm satisfies the termination condition:
\begin{lemma}[Termination]\label{thm:termination}
    The algorithm of Fig.~\ref{alg:safe_agreement} satisfies the termination
    condition.
\end{lemma}

\begin{lemma}[Nontriviality]\label{thm:nontriviality}
    The algorithm of Fig.~\ref{alg:safe_agreement} satisfies the nontriviality condition.
\end{lemma}

\begin{proof}
    Fix an execution, in which more than one \propose{} operation are performed
    and no process fails while performing these \propose{} operations.
    Define $V_{i}$ in the same way as in Lemma~\ref{thm:inclusion}.
    It is enough to show that $V_i\neq\emptyset$ for all $i=0,\ldots,n-1$.
    We prove this by induction on $i$.

    \emph{Basis}: 
        Every process that performs a \propose{} operation adds its argument to 
        $\weakset[0]$ and thus $V_0\neq\emptyset$.

    \emph{Induction step}:
        Suppose that $V_i\neq\emptyset$ holds.
        If $\#V_i<2$, every process that executes the $i$-th iteration of the for loop
        proceeds to the $(i+1)$-th iteration because the condition ``$\#Snap\ge 2$'' 
        at Line~\ref{line:if} is never satisfied.
        If $\#V_i\ge2$, the process that adds $\max V_i$ to $\weakset[i]$ never
        satisfies the condition of the if statement and proceeds to the $(i+1)$-th iteration,
        because $\max V_i$ cannot be the minimum of any subset of $V_i$ of cardinality
        more than 1.
        In either case, $V_{i+1}\neq\emptyset$.
\end{proof}

By Lemmas~\ref{thm:validity}, \ref{thm:agreement}, \ref{thm:termination}, and 
\ref{thm:nontriviality}, we obtain the following theorem.

\begin{theorem}
    The algorithm of Fig.~\ref{alg:safe_agreement} is an anonymous wait-free
    implementation of safe agreement object.
\end{theorem}

Note that the space complexity of the implementation of the safe agreement object
is $n$ atomic registers, because an arbitrary finite number of atomic weak set objects
are simulated on top of a single atomic weak set object.

% !TEX root =  ./paper.tex
\section{\boldmath $t$-Resilient Solvable Colorless Tasks}
\label{sec:charact}

We now give a characterization of the $t$-resilient solvability of colorless tasks
in the anonymous model. 

\begin{theorem}\label{thm:equivalence}
    A colorless task is $t$-resilient solvable by $n$ anonymous processes
    with atomic weak set objects if and only if it is $t$-resilient solvable 
    by $n$ non-anonymous processes with atomic snapshot objects.
    Moreover, if a colorless task is $t$-resilient solvable by
    $n$ anonymous processes, it is solvable with $n$ atomic registers.
\end{theorem}

The only if part of the theorem is immediate because every anonymous protocol 
can be executed by non-anonymous processes.
We prove the if part by two different approaches, topological one and operational one.

\subsection{Topological Approach}

We prove the if part of Theorem~\ref{thm:equivalence} by a topological argument.

We first show that the $(t+1)$-set agreement is $t$-resilient solvable
by $n$ anonymous processes.
An algorithm of Fig.~\ref{alg:set_agreement} presents an anonymous $t$-resilient protocol
for the $(t+1)$-set agreement.
In the protocol, each process first proposes its input value to \safeagree$[i]$
for $i=0,\ldots,n-1$.
Then, the process repeatedly performs a \resolve{} operation to all \safeagree$[i]$
in the round-robin manner until it gets non-$\perp$ value.
Once the process gets non-$\perp$ value, the process returns the value.

\begin{figure}[h]
\hrule \vspace{1mm} {%\small
\setcounter{linenumber}{0}
\begin{tabbing}
bbb\=bbb\=bbb\=bbb\=bbb\=bbb\=bbb\=bbb\=bbb\=bbb\=bbb\=bbb\=  \kill

Shared variable : \\
\>\texttt{array of safe agreement objects} : $\safeagree[0\ldots t]$ \\[2mm]

\textsc{Code for a process} \\[2mm] 
Local variable:\\
    \>\texttt{Integer} $i$ init 0\\
    \>\texttt{Value} $result$ init $\perp$\\[2mm]

\setagree{v}:\\
\nnll\> \textbf{for} $i=0,\ldots,t$ \textbf{do}\\
\nnll\>\>   $\safeagree[i].\propose{v}$\\
\nnll\>     $i=0$\\
\nnll\label{line:while}\> \textbf{while} $result=\perp$ \textbf{do}\\
\nnll\>\>   $result=\safeagree[i].\resolve{}$\\
\nnll\label{line:ewhile}\>\>   $i=i+1\mod{t+1}$\\
\nnll\> \textbf{return} $result$\\
\end{tabbing}
\vspace{-6mm}
\hrule}
\caption{Anonymous $t$-resilient $(t+1)$-set agreement protocol}
\label{alg:set_agreement}
\end{figure}

\begin{theorem}
    The algorithm of Fig.~\ref{alg:set_agreement} is a $t$-resilient anonymous
    protocol for the $(t+1)$-set agreement.
\end{theorem}

\begin{proof}
    \emph{Termination:} 
    In the protocol, each process performs \propose{} operations
    to $\safeagree[0],\ldots,\safeagree[n-1]$ sequentially. 
    Thus, even if $t$ processes fail,
    there is at least one safe agreement object such that 
    no process fails while performing a \propose{} operation on the object.
    By the nontriviality property of safe agreement objects, 
    after some time instance, \resolve{} operations on some safe agreement object
    return non-$\perp$ value and thus the while loop of 
    Line~\ref{line:while}--\ref{line:ewhile}
    eventually terminates. 
    
    \emph{Validity:} 
        Every argument of a \propose{} operation is a proposed  value.
        Because of the validity property of safe agreement objects,
        a non-$\perp$ value returned by some \resolve{} operation is
        one of the arguments of \propose{} operations.
        Thus, the validity condition holds.
        
    \emph{$k$-Agreement:}
         There are $t+1$ distinct safe agreement objects.
         Thus, by the agreement property of safe agreement objects,
         at most $t+1$ distinct values are decided.
\end{proof}

As the $b$-iterated barycentric agreement task is wait-free solvable 
by anonymous processes~\cite{yanagisawa2017wait}, the following theorem holds.

\begin{lemma}\label{thm:top_if}
    Let $T=(\cI,\cO,\Delta)$ be a colorless task.
    If there exists a continuous map $f:|\skel^{t}\cI|\to|\cO|$ carried by $\Delta$,
    $T$ is $t$-resilient solvable by $n$ anonymous processes.
\end{lemma}

\begin{proof}
    By Lemma~\ref{thm:approximation}, there is an integer $b$ and a simplicial map 
    $\delta:\bary^b\skel^t\cI\to\cO$ that satisfies
    $\delta(\bary^b \sigma)\subseteq\Delta(\sigma)$ for every $\sigma\in \skel^t \cI$.

    The following anonymous protocol solves the colorless task.
    Suppose that the set of all inputs to the processes is $s\in \cI$.
    Execute first  the anonymous $(t+1)$-set agreement protocol,
    and then the $b$-iterated barycentric agreement protocol (for sufficiently large value of $b$). 
    Each process chooses a vertex of $\bary^{b}\skel^{t}\sigma$.
    Finally, each process determines its output by applying $\delta$ to the vertex it chose.
\end{proof}

The if part of Theorem~\ref{thm:equivalence} follows from Lemma~\ref{thm:top_if}
and the following theorem by Herlihy and Rajsbaum:

\begin{theorem}[{\cite[Theorem 4.3]{herlihy2010topology}}]\label{thm:top_char_epo}
    A colorless task $T=(I,O,\Delta)$ is $t$-resilient solvable 
    by $n$ non-anonymous processes if and only if there exists a continuous map
    $f:|\skel^{t}I|\to|O|$ carried by $\Delta$.
\end{theorem}

Note that the protocol appeared in the proof of Lemma~\ref{thm:top_if} only makes use of 
a finite number of atomic weak set objects, which are constructed on top of a single 
atomic weak set object.
Thus, every colorless task that is $t$-resilient solvable by $n$ anonymous processes
is solved with $n$ atomic registers.
The space complexity lower bound of Theorem~\ref{thm:equivalence} follows.

\subsection{Simulation-Based Approach}

We now prove the if part of Theorem~\ref{thm:equivalence} by a simulation, 
which is an anonymous variant of the BG-simulation~\cite{BorowskyGLR01}.
More precisely, we show that $n$ anonymous $t$-resilient processes 
with atomic weak set objects can simulate $n$ non-anonymous $t$-resilient processes 
with \emph{atomic snapshot objects}.
We write anonymous simulators by $p_0$, $\ldots$, $p_{n-1}$ 
and non-anonymous simulated processes by $P_0$, $\ldots$, $P_{n-1}$.
Without loss of generality, we may assume that non-anonymous processes communicate via
a single $n$-ary atomic snapshot object and execute a \emph{full-information protocol}.
In the protocol, the process $P_i$ repeatedly writes its local state to the $i$-th 
component of the array, takes a snapshot of the whole array and update its state 
by the result of the snapshot until it reaches a termination state.
When the process reaches the termination state, it decides on the value obtained by
applying some predefined function $f$ to the state.

Our simulation algorithm for each simulator is presented in Fig.~\ref{alg:simulation}.
The algorithm makes use of a two dimensional array of safe agreement object 
$\safeagree[0\ldots][0\ldots n-1]$, where the column $\safeagree[0\ldots][i]$
is for storing simulated states of the process $P_i$.
The local variables $\mathit{view}_i$ and $\mathit{round}_i$ stand for the current 
simulated state and the current simulated round of $P_i$ respectively.
The function $\mathit{latest\_views}$ maps a set of tuples consists of
a process name, its simulated state, and its simulated round to
the array whose $i$-th component is the simulated view of $P_i$
associated with the largest simulated round number of $P_i$.
The function $\mathit{latest\_round}_i$ maps a set of the same kind to
the latest round number of $P_i$.

In the algorithm, each simulator first proposes its input value to $\safeagree[0][i]$ 
for all $P_0,\ldots,P_{n-1}$.
Then, the simulator repeats the following procedure for $P_0,\ldots,P_{n-1}$
in the round-robin manner until one of $P_0,\ldots,P_{n-1}$ reach a termination state:
it performs \resolve{} operation on $\safeagree[\mathit{round}_i][i]$;
if the return value of the \resolve{} operation is not $\perp$, 
the simulator adds the return value, with the name $P_i$ and its current simulated round,
to $\weakset{}$,
updates simulated state and round,
and proposes the new simulated state of $P_i$ to $\safeagree[\mathit{round}_i][i]$.
\begin{figure}[h]
\hrule \vspace{1mm} {%\small
\setcounter{linenumber}{0}
\begin{tabbing}
bbb\=bbb\=bbb\=bbb\=bbb\=bbb\=bbb\=bbb\=bbb\=bbb\=bbb\=bbb\=  \kill

Shared variable : \\
    \>\texttt{atomic weak set} : $\weakset$\\
    \>\texttt{array of safe agreement objects} : $\safeagree[0\ldots][0\ldots n-1]$\\[2mm]

\textsc{Code for a process} \\[2mm] 
Local variable:\\
    \>\texttt{Value} $\mathit{view}_i$ init $\perp$ for $i=0,\ldots,n-1$\\
    \>\texttt{Integer} $\mathit{round}_i$ init $0$ for $i=0,\ldots,n-1$\\
    \>\texttt{Integer} $i$ init 0\\
    \>\texttt{Value} $\mathit{snap}$ init $\perp$\\[2mm]

Simulation($v$):\\
\nnll\> \textbf{for} $i=0,\ldots,n-1$ \textbf{do}\\
\nnll\>\>   $\safeagree[0][i].\propose{v}$\\
\nnll\> \textbf{while} true \textbf{do}\\
\nnll\>\>    \textbf{for} $i=0,\ldots,n-1$ \textbf{do} \\
\nnll\>\>\>     $\mathit{view}_i=\safeagree[\mathit{round}_i][i].\resolve{}$\\
\nnll\>\>\>     \textbf{if} $\mathit{view}_i$ is a termination state of $P_i$ \textbf{then}\\
\nnll\>\>\>\>       \textbf{return} $f(\mathit{view}_i)$\\
\nnll\>\>\>     \textbf{elseif} $\mathit{view}_i\neq\perp$ \textbf{then}\\
\nnll\>\>\>\>       $\weakset.\add{(P_i,\mathit{view}_i,\mathit{round}_i)}$\\
\nnll\>\>\>\>       $\mathit{snap}=\weakset.\get{}$\\
\nnll\>\>\>\>       $\mathit{view}_i=\mathit{latest\_views}(\mathit{snap})$\\
\nnll\>\>\>\>       $\mathit{round}_i=\mathit{latest\_round}_i(\mathit{snap})+1$\\ 
                    % Note: round_i=round_i+1 also works.
\nnll\>\>\>\>       $\safeagree[\mathit{round}_i][i].\propose{\mathit{view}_i}$\\
\end{tabbing}
\vspace{-6mm}
\hrule}
\caption{$n$ anonymous processes simulates $n$ non-anonymous processes}
\label{alg:simulation}
\end{figure}

By the use of safe agreement objects, simulators can agree on the return value 
of each simulated snapshot.
Note that there is no need to use a safe agreement object on each simulated update
because each value to be updated is deterministically determined by the return value
of the preceding simulated snapshot.
In the algorithm of Fig.~\ref{alg:simulation}, each simulator performs
\propose{} operations sequentially.
Thus, even though $t$ simulators crash, they block at most $t$ simulated processes
by the nontriviality property of the safe agreement object.
By these observations, we establish the following lemma:

\begin{lemma}
    If a colorless task $t$-resilient solvable by $n$ non-anonymous processes
    with atomic snapshot objects,
    it is also $t$-resilient solvable by $n$ anonymous processes 
    with atomic weak set objects.
\end{lemma}
The proof of the lemma is  similar to the proof of Theorem 5 in~\cite{BorowskyGLR01}.

The space complexity of the simulation of Fig.~\ref{alg:simulation} is
exactly $n$ atomic registers because a single atomic weak set object can simulate
an arbitrary finite number of atomic weak set objects and safe agreement objects
in the non-blocking manner.
This establishes the space complexity lower bound of Theorem~\ref{thm:equivalence}.

%!TEX root = paper.tex

\section{Conclusion}

The theory of distributed computing in~\cite{herlihy2013distributed} has been
successful in characterizing task computability in a variety of shared-memory, message
passing and mobile robot models using  combinatorial topology.  
For the case of shared-memory models, it assumes 
that the processes, $p_0,\ldots,p_{n-1}$, 
communicate using SWMR registers.
Recently a characterization of wait-free colorless task solvability has been derived 
for  the anonymous case, where processes have no identifiers and communicate through multi-writer/multi-reader registers~\cite{yanagisawa2016wait-free}. 
In this paper we have extended this result to the case where at most $t$ processes may crash, $1\leq t<n$. 
Furthermore, we have shown that any $t$-resilient solvable colorless task can be $t$-resilient 
solvable anonymously using only $n$ MWMR registers.

Some of the avenues for future research are the following.
It would be interesting to look for lower bound on the number of MWMR registers needed
to solve specific colorless tasks. Also, to investigate which non-colorless tasks are
solvable in the anonymous setting.
We have derived our result through a series of reductions that seem interesting in themselves,
to study further anonymous computability, especially for long-lived objects (as opposed to tasks)
and uniform solvability (instead of a fixed number of processes $n$).
For this, it may be useful to extend our non-blocking implementation of the weak set object
to be wait-free. Also, to eliminate our assumption of finite  input complexes. 
%finite number of subdivisions, and in weak set Ag

\section*{Acknowledgement}
We would like to thank Petr Kuznetsov and Susumu Nishimura for helpful discussions
related to this paper. 

\bibliographystyle{plain}
\bibliography{references}

\begin{thebibliography}{10}

\bibitem{angluin1980local}
Dana Angluin.
\newblock Local and global properties in networks of processors.
\newblock In {\em Proceedings of the 12th Annual ACM Symposium on Theory of
  Computing (STOC)}, pages 82--93, 1980.

\bibitem{attiya2006adapting}
Hagit Attiya.
\newblock {Adapting to Point Contention with Long-Lived Safe Agreement}.
\newblock In Paola Flocchini and Leszek G{\k{a}}sieniec, editors, {\em 13th
  Int. Conf. Structural Information and Communication Complexity (SIROCCO)},
  volume 4056 of {\em Lecture Notes in Computer Science (LNCS)}, pages 10--23.
  Springer, 2006.

\bibitem{attiya2002computing}
Hagit Attiya, Alla Gorbach, and Shlomo Moran.
\newblock Computing in totally anonymous asynchronous shared memory systems.
\newblock {\em Information and Computation}, 173(2):162--183, 2002.

\bibitem{baldoni2010value-based}
Roberto Baldoni, Silvia Bonomi, and Michel Raynal.
\newblock Value-based sequential consistency for set objects in dynamic
  distributed systems.
\newblock In {\em Proceedings of the 16th International Euro-Par Conference
  (Euro-Par)}, volume 6271 of {\em Lecture Notes in Computer Science (LNCS)},
  pages 523--534, 2010.

\bibitem{BorowskyGLR01}
Elizabeth Borowsky, Eli Gafni, Nancy Lynch, and Sergio Rajsbaum.
\newblock The {BG} distributed simulation algorithm.
\newblock {\em Distributed Computing}, 14(3):127--146, 2001.

\bibitem{bouzid2016anonymity}
Zohir Bouzid and Corentin Travers.
\newblock {Anonymity-Preserving Failure Detectors}.
\newblock In Cyril Gavoille and David Ilcinkas, editors, {\em 30th
  International Symposium Distributed Computing (DISC)}, volume 9888 of {\em
  Lecture Notes in Computer Science (LNCS)}, pages 173--186. Springer, 2016.

\bibitem{CapdevielleAnonAg2017}
Claire Capdevielle, Colette Johnen, Petr Kuznetsov, and Alessia Milani.
\newblock On the uncontended complexity of anonymous agreement.
\newblock {\em Distributed Computing}, 30(6):459--468, Dec 2017.

\bibitem{CastanedaRenam:2011}
Armando Casta\~{n}eda, Sergio Rajsbaum, and Michel Raynal.
\newblock The renaming problem in shared memory systems: An introduction.
\newblock {\em Comput. Sci. Rev.}, 5(3):229--251, August 2011.

\bibitem{Chaudhuri1993more}
S.~Chaudhuri.
\newblock More choices allow more faults: Set consensus problems in totally
  asynchronous systems.
\newblock {\em Information and Computation}, 105(1):132 -- 158, 1993.

\bibitem{delporte-gallet2009two}
Carole Delporte-Gallet and Hugues Fauconnier.
\newblock Two consensus algorithms with atomic registers and failure detector
  $\omega$.
\newblock In {\em Proceedings of the 10th International Conference on
  Distributed Computing and Networking (ICDCN)}, volume 5408 of {\em Lecture
  Notes in Computer Science (LNCS)}, pages 251--262, 2009.

\bibitem{DFGRbootstrapTCS15}
Carole Delporte-Gallet, Hugues Fauconnier, Eli Gafni, and Sergio Rajsbaum.
\newblock Linear space bootstrap communication schemes.
\newblock {\em Theoretical Computer Science}, 561(Part B):122 -- 133, 2015.
\newblock Special Issue on Distributed Computing and Networking.

\bibitem{Ellen2008}
Faith Ellen, Panagiota Fatourou, and Eric Ruppert.
\newblock The space complexity of unbounded timestamps.
\newblock {\em Distributed Computing}, 21(2):103--115, Jul 2008.

\bibitem{fischer1985impossibility}
Michael~J. Fischer, Nancy~A. Lynch, and Michael~S. Paterson.
\newblock Impossibility of distributed consensus with one faulty process.
\newblock {\em J. ACM}, 32(2):374--382, April 1985.

\bibitem{gafni2009extended}
Eli Gafni.
\newblock The extended bg-simulation and the characterization of t-resiliency.
\newblock In {\em Proceedings of the Forty-first Annual ACM Symposium on Theory
  of Computing}, STOC '09, pages 85--92, New York, NY, USA, 2009. ACM.

\bibitem{GelashviliAnonCons2015}
Rati Gelashvili.
\newblock {\em On the Optimal Space Complexity of Consensus for Anonymous
  Processes}, pages 452--466.
\newblock Springer Berlin Heidelberg, Berlin, Heidelberg, 2015.

\bibitem{guerraoui2007anonymous}
Rachid Guerraoui and Eric Ruppert.
\newblock Anonymous and fault-tolerant shared-memory computing.
\newblock {\em Distributed Computing}, 20(3):165--177, 2007.

\bibitem{herlihy2013distributed}
Maurice Herlihy, Dmitry Kozlov, and Sergio Rajsbaum.
\newblock {\em Distributed computing through combinatorial topology}.
\newblock Morgan Kaufmann, 2013.

\bibitem{herlihy1997decidability}
Maurice Herlihy and Sergio Rajsbaum.
\newblock The decidability of distributed decision tasks (extended abstract).
\newblock In {\em Proceedings of the 29th Annual ACM Symposium on Theory of
  Computing (STOC)}, pages 589--598, 1997.

\bibitem{herlihy2003classification}
Maurice Herlihy and Sergio Rajsbaum.
\newblock A classification of wait-free loop agreement tasks.
\newblock {\em Theoretical Computer Science}, 291(1):55 -- 77, 2003.

\bibitem{herlihy2010topology}
Maurice Herlihy and Sergio Rajsbaum.
\newblock The topology of shared-memory adversaries.
\newblock In {\em Proceedings of the 29th ACM SIGACT-SIGOPS Symposium on
  Principles of Distributed Computing (PODC)}, pages 105--113, 2010.

\bibitem{HerlihyR12}
Maurice Herlihy and Sergio Rajsbaum.
\newblock {Simulations and reductions for colorless tasks}.
\newblock In {\em Proceedings of the 2012 ACM symposium on Principles of
  distributed computing}, PODC '12, pages 253--260, New York, NY, USA, 2012.
  ACM.

\bibitem{herlihy2017computingJ}
Maurice Herlihy, Sergio Rajsbaum, Michel Raynal, and Julien Stainer.
\newblock From wait-free to arbitrary concurrent solo executions in colorless
  distributed computing.
\newblock {\em Theor. Comput. Sci.}, 683:1--21, 2017.

\bibitem{herlihy1999topological}
Maurice Herlihy and Nir Shavit.
\newblock The topological structure of asynchronous computability.
\newblock {\em J. ACM}, 46(6):858--923, November 1999.

\bibitem{herlihy1990linearizability}
Maurice~P. Herlihy and Jeannette~M. Wing.
\newblock Linearizability: A correctness condition for concurrent objects.
\newblock {\em ACM Trans. Program. Lang. Syst.}, 12(3):463--492, July 1990.

\bibitem{jayanti1991wakeup}
Prasad Jayanti and Sam Toueg.
\newblock Wakeup under read/write atomicity.
\newblock In {\em Proceedings of the 4th International Workshop on Distributed
  Algorithms}, pages 277--288, 1991.

\bibitem{gsb2016}
Armando~Casta\ {n}eda, Damien Imbs, Sergio Rajsbaum, and Michel Raynal.
\newblock Generalized symmetry breaking tasks and nondeterminism in concurrent
  objects.
\newblock {\em SIAM Journal on Computing}, 45(2):379--414, 2016.

\bibitem{robotsIPDPS17}
S.~Rajsbaum, A.~Castañeda, D.~F. Peñaloza, and M.~Alcántara.
\newblock Fault-tolerant robot gathering problems on graphs with arbitrary
  appearing times.
\newblock In {\em 2017 IEEE International Parallel and Distributed Processing
  Symposium (IPDPS)}, pages 493--502, May 2017.

\bibitem{saraph2016asynchronous}
Vikram Saraph, Maurice Herlihy, and Eli Gafni.
\newblock {\em Asynchronous Computability Theorems for t-Resilient Systems},
  pages 428--441.
\newblock Springer Berlin Heidelberg, Berlin, Heidelberg, 2016.

\bibitem{yanagisawa2016wait-free}
Nayuta Yanagisawa.
\newblock Wait-free solvability of colorless tasks in anonymous shared-memory
  model.
\newblock In {\em Proceedings of the 18th International Symposium on
  Stabilization, Safety, and Security of Distributed Systems (SSS)}, pages
  415--429, 2016.

\bibitem{yanagisawa2017wait}
Nayuta Yanagisawa.
\newblock Wait-free solvability of colorless tasks in anonymous shared-memory
  model.
\newblock {\em Theory of Computing Systems}, 2017 (to appear).

\end{thebibliography}

%!TEX root = paper.tex

\newpage
\section{Appendix}

%lemma from 1 ...observation
\setcounter{lemma}{0}
\setcounter{theorem}{0}

%
% We briefly present how a MWMR atomic registers are simulated on top of 
% an atomic weak set object in the non-blocking manner.
% This establishes a sort of equivalence between atomic registers and atomic weak set objects.
%
% We simulate a register $R$ on an atomic weak set \weakset{}.
% Let $V$ be a set of values to be added to $R$ and assume that $V$ is totally ordered.
% Suppose that process may only add a tuple of a value $v\in V$ and a natural number.
% We call the second element of the tuple `tag'.
%
% To simulate a write of $v$ to $R$, a process first performs a \get{} operation 
% on \weakset{} to obtain the maximum tag, denoted by $\mathit{tag_{max}}$, and then 
% add $(v, \mathit{tag_{max}}+1)$ to \weakset.
% %
% To simulate a read from $R$, a process performs a \get{} operation on \weakset{}
% and chooses the value with the $\mathit{tag_{max}}$.
% If there are more than one values with $\mathit{tag_{max}}$, the process choose
% the largest tuple with respect to the lexicographic order.
%
% The above simulation is a almost verbatim copy of the simulation of MWMR atomic registers
% on SWMR atomic registers in the non-anonymous setting.
% Thus, see \S 13.4.5 of \cite{lynch1996distributed} for the proof.
%

\subsection{Correctness Proofs of Weak Set Implementation}

We give the correctness proofs of the weak free implementation.

\paragraph{Safety}
%%% MODEL
Given an operation $op$,  $invoc (op)$ denotes its invocation and
$resp(op)$ its response.

% History $H$; sequential history $H_{seq}$
% %%%%%

Let $H$ be an history of the algorithm. $H_{seq}$ denotes the sequential
history in which each operation of $H$ appears as if it has been
executed at a single point (the linearization) of time line.
We have to define linearization points and prove that:
\begin{itemize}
    \item
        the linearization point of each operation $\textsc{get}$ and 
        $\textsc{add}$ appear  between the beginning and the end of this operation, and 
    \item
        the sequential history that we get with these points respect the
        sequential specification of the weak set. 
\end{itemize}
Most of the details of the proofs are in the appendix.

Consider an history $H$, let $v$ be  a value or a set of values,
define time  $\tau_{v} $ as the first time, if any, that $v$ belongs to
all registers of $R$. When there is no such time,  $\tau_v$ is $\bot$.

\begin{lemma}(Lemma \ref{lemma:terminaison})\label{lemmaA:terminaison}
    If the operation $\textsc{add}(v)$ terminates then before the end of
    this operation $v$ belongs to all registers of $R$. 
    If the operation $\textsc{get}()$ terminates  and returns $V$ then
    before the end of this operation $V$ belongs to all registers of $R$. 
\end{lemma}
\begin{proof}
    If the $ \textsc{add}(v)$ terminates then the condition
    Line~\ref{li:while-add} is false.  
    $v$ belongs to all cells of $Snap$. $Snap$ comes from the snapshot
    Line~\ref{li:scan1} or  Line~\ref{li:scan2}. 
    When the process executes this line, $v$ belongs to all registers of $R$.

    If the $ \textsc{get}()$ terminates and returns $V$ then the
    condition Line~\ref{li:while-get} is false.  
    $V$ belongs to all cells of $Snap$. $Snap$ comes from the snapshot
    Line~\ref{li:scan3} or  Line~\ref{li:scan4}. 
    When the process executes this line, $V$ belongs to all registers of $R$.
\end{proof}

We prove that the algorithm is non-blocking, namely, if processes
perform operations forever, an infinite number of  operations
terminates.

% \begin{Obs}\label{obs:increasing-view}
% Considering the inclusion $\subseteq$, for each process variable $View$ is not decreasing
% \end{Obs}
By contradiction, assume that there is only a finite number of
operations $\textsc{get}$ and $\textsc{add}$ and some
operations made by correct processes do not terminate. 

Operations $\textsc{add}$ or $\textsc{get}$ may not terminate because the 
termination conditions of the while loop are not satisfied (Lines~\ref{li:while-add}
or~\ref{li:while-get}): for an $\textsc{add}(v)$ operation, in each $scan$ made by 
the process, $v$ is not in at least one of the registers of $R$, and for a $\textsc{get}()$ 
operation, in each $snap$, all the registers are not equal to the view of the process.

There is a time $\tau_0$ after which there is no new process crash
and all processes that terminate $\textsc{get}$ or $\textsc{add}$ operations in the
run have already terminated. Consider the set $N$ of processes alive
after time $\tau_0$ that do not terminate operations in the run. Note
that after time $\tau_0$ only processes in $N$ take steps and 
as no process in $N$ may crash each process in $N$ makes an
infinite number of steps. 

From observation~\ref{obsA:integrity} if there is a finite number of
operations, then all variables $View$ are subsets of a finite set of values.
Moreover from observation~\ref{obsA:increasing-view} the views of each
process are increasing, then there is a time
$\tau_1>\tau_0$ after which the view of each process $p$ in $N$ converges
to a \emph{stable view} $SView_p$: after time $\tau_1$ forever  the view of $p$
is $SView_p$.
In the following $SV$ denotes $\{SView_p| p\in N\}$ the set of all stable
views for processes in $N$. Observe that:
\begin{observation}( Observation \ref{obs:vinView}) \label{obsA:vinView}
    If  $p$ does not terminate an $\textsc{add}(v)$ then $v\in SView_p$.
\end{observation}

After time $\tau_1$ processes only update $R$ with their stables views $SView$, 
and as each process in $N$ updates infinitely often each register of $R$ with 
its $SView$ there is a time $\tau_2$ after which all registers in $R$ contain 
only stable views of processes in $N$:
\begin{observation}(Observation \ref{obs:st2})\label{obsA:st2})
    After time $\tau_2$ for all $i$, $R[i]\in SV$.
\end{observation}

\paragraph{liveness}
Among stable views consider any minimal view $SView_0$ for inclusion,
i.e. for all $S\in SV$, $S\subseteq SView_0$ implies $SView_0=S$.

Consider any process $p\in N$ having the $SView_0$ as stable view, eventually 
$p$ makes a scan of the memory $R$ (Line~\ref{li:scan1} or~\ref{li:scan2} for 
$\textsc{add}$,  Line~\ref{li:scan3} or~\ref{li:scan4} for $\textsc{get}$). 
Let $Snap$ be the array returned by the $scan$. $Snap$ is the value of the 
array of registers $R$ at some time after $\tau_2$. 
Then $p$ adds $\bigcup_{1\leq i\leq n}  Snap[i]$ to its view $SView_0$.
The $SView_0$ being stable we have $\bigcup_{1\leq i\leq n} Snap[i]\subseteq SView_0$ 
and then for all $i$, $Snap[i] \subseteq   SView_0$.
But by observation~\ref{obsA:st2}, $R[i]=Snap[i]$ is a stable view $S \in SV$. 
Then by the minimality of $SView_0$,  for all $i$ we have $Snap[i]=SView_0$.
Then consider the two following cases:
\begin{itemize}
    \item
        if $p$ is performing an $\textsc{add}(v)$, as $v \in SView_p=SView_0$, 
        for all $i$, $v\in Snap[i]$ and the loop condition $(\# \{r |v $ in
        $ Snap[r]\} <  n)$ (Line~\ref{li:while-add}) is false and $p$
        terminates $\textsc{add}(v)$-- A
        contradiction
    \item
        if $p$ is performing an $\textsc{get}()$, then the loop continuation condition 
        $(\#\{r |View= Snap[r]\} <   n)$ is false and $p$ terminates
        operation $\textsc{get}()$ --- A contradiction
\end{itemize}
We deduce that there is no minimal stable view proving that $SV=\emptyset$
and also $N=\emptyset$.
\paragraph
{\it Linearization points for operations $\textsc{add}$ and $\textsc{get}$}:

By Lemma~\ref{lemmaA:terminaison}, $\tau_v\neq \bot$ for each operation
$\textsc{add}(v)$ that terminates  
and $\tau_V\neq \bot$ for each operation $\textsc{get}()$ that
terminates  and returns $V$.

If $\tau_v\neq \bot$, the linearization point $\tau_{op}$ of an
operation $op=\textsc{add}(v)$ is $max\{ \tau_{v}, invoc(op) \}$.     
If $\tau_v=\bot$, the operation $op$ does not terminate and   
is not linearized.

The linearization point $\tau_{op}$ of an operation $op=\textsc{get}()$ 
that returns $V$ is $max\{ \tau_{V}, invoc(op) \}$. A $\textsc{get}()$ 
operation that does not terminate is not linearized. 

Directly from the definition and from Lemma\ref{lemmaA:terminaison},
the linearization point $\tau_{op}$ occurs between the invocation and 
the response of $op$:
\begin{lemma}
    Let $op$ be an operation of $H$. If $\tau_{op}$ is defined then $
    \tau_{op}$ belongs to $[invoc(op), resp(op)]$. 
    If $\tau_{op}$ is undefined then $op$ does not terminate and is not linearized.
\end{lemma}

In the following we consider the next snapshot operation on $R$ of each process.
This next snapshot operation  is either  a $scan$ or an $update$ or there is 
no next snapshot operation (when the process is about to satisfy the termination 
loop condition: $card\{r|v \in Snap[r] \}\geq n$ for an $add$ and 
$card\{r|View= Snap[r] \}\geq n$ for an $update$).

\noindent
Consider any time $\tau$. Define $r_v(\tau)$, $w_v(\tau)$, $c_v(\tau)$
and $\alpha_v(\tau)$:
\begin{itemize}
    \item
        $r_v(\tau)$ is the number of processes for which, after time
        $\tau$, the next snapshot operation is a $scan$. 
    \item
        $w_v(\tau)$ is the number of processes such that (1) $v \in View$ at
        time $\tau$ and for which after time $\tau$ the next snapshot operation 
        is an $update$,  or (2)
        %its $View$ contains $v$ and 
        there is no next snapshot operation for that process 
        (the process has finished -or is going to finish- its main loop or it takes no more steps).
    \item
        $c_v(\tau)$ is the number of registers that contains $v$ at time $\tau$.
    \item
        $\alpha_v(\tau)$ is defined: $\alpha_v(\tau)=  r_v(\tau)+w_v(\tau) +c_v(\tau)$ 
\end{itemize}
We make first some easy observations.

\noindent
For a process the next snapshot operation  is either a $snap$ or an $update$ or nothing:
\begin{observation}
    %(Observation \ref{obs:trivial})
    \label{obsA:trivial}
    $r_v(\tau)+w_v(\tau)\leq n$.
\end{observation}
\begin{observation}
    %(Observation \ref{obs:integrity})
    \label{obsA:integrity}
    All values in  variables $View$ have been proposed by some $\textsc{get}$
\end{observation}
\begin{observation}
    %(Observation \ref{obs:increasing-view})
    \label{obsA:increasing-view}
    Considering the inclusion $\subseteq$, for each process, variable $View$ is not decreasing
\end{observation}
\noindent
%From the algorithm:
\begin{observation}
    %( Observation \ref{obs:updatescan}) 
    \label{obsA:updatescan}
    Each $update$ made by a process is followed for this process by a
    $scan$ or the process stops to take step.
\end{observation}
\noindent

% \begin{Obs}\label{obs:scan}
%Each $scan$ made by a process is followed for this process by an
%$update$ or there is no next
%snapshot operation for that process.
%\end{Obs}
\begin{observation}
    $\alpha_v(\tau)$ may only be modified by a $scan$ (Lines~\ref{li:scan1},
    \ref{li:scan2}, \ref{li:scan3} and \ref{li:scan4}) or an $update$ (Lines
    \ref{li:update1} and \ref{li:update2}).
\end{observation}

Due to this observation in the following we consider only steps of
processes that are $scan$ or $update$.

\begin{lemma}(Lemma \ref{lemma:increasingstep})\label{lemmaA:increasingstep}
    Assume $\alpha_v(\tau) >n$ and the next step in $H$, is made at
    time $\tau'\geq \tau$, then we have  $\alpha_v(\tau)\leq \alpha_v(\tau')$.
\end{lemma}
\begin{proof}
    Consider $H$ an history and some time $\tau$ and 
    assume that at time $\tau$,  we have $\alpha_v(\tau) >n$.
    Consider the next snapshot operation in $H$ in algorithm of Fig.~\ref{fig:aws}.
    Let $p$ the process that executes this operation.
    % $\alpha_v$ may be  modified only (a) when a process for which
    % the next step is $Scan$ makes that scan (Lines~\ref{li:scan1},
    % \ref{li:scan2}, \ref{li:scan3} and \ref{li:scan4}) 
    % (b) 
    % when a process for which
    % the next step is $update()$ makes that $update$ (Lines
    % \ref{li:update1} and \ref{li:update2})
    \begin{itemize}
        \item
            The next step is a $scan$ (Lines~\ref{li:scan1},
            \ref{li:scan2}, \ref{li:scan3} and \ref{li:scan4}) then by 
            observation~\ref{obsA:trivial} and the hypothesis $\alpha_v(\tau) >n$, 
            we have $c_v(\tau) \geq 1$, hence at least one of the element of array $R$ 
            contains $v$ and the result of the $scan$ contains $v$.  So after this $scan$, 
            $View$ at process $p$ contains v.

            Consider now the next snapshot operation of that process either it will make 
            an update and this next update will contain $v$  or the process will take 
            no more snapshot operations   or the process will make another $scan$.
            In the two first cases , $w_v(\tau')= w_v(\tau)+1$, and $r_v(\tau')=r_v(\tau)-1$.
            Hence at time $\tau'$ we have $w_v(\tau')= w_v(\tau)+1$, $c_v(\tau')=c_v(\tau)$ and 
            $r_v(\tau')=r_v(\tau)-1$. 
            Then $$\alpha_v(\tau')=r_v(\tau)-1+w_v(\tau) +1+c_ v(\tau) =\alpha_v(\tau)$$.
            In the last case $w_v(\tau')= w_v(\tau)$, and $r_v(\tau')=r_v(\tau)$.
            Hence at time $\tau'$ we have  $c_v(\tau')=c_v(\tau)$. 
            Then $$\alpha_v(\tau')=r_v(\tau)+w_v(\tau) +c_ v(\tau) =\alpha_v(\tau)$$

        \item
            The next step is an $update$ (Lines \ref{li:update1} and \ref{li:update2}).
            By Observation~\ref{obsA:updatescan}, either (1) the next snapshot operation 
            of $p$  is a $scan$  then $r_v(\tau')=r_v(\tau)+1$, or 
            (2) there is no next snapshot operation for $p$ then $r_v(\tau')=r_v(\tau)$.

            Consider the first case, and the two following subcases:
            \begin{itemize}
                \item
                    That $update$ is an $update(-,V)$,  with $v\in V$ or $v \subseteq V$, 
                    then $w_v(\tau')= w_v(\tau)-1$, as the update write $V$ in $R$,
                    $c_v(\tau')\geq c_v(\tau)$ and by Observation~\ref{obsA:updatescan},
                    $r_v(\tau')=r_v(\tau)+1$. Hence:
                    $$\alpha_v(\tau')\geq r_v(\tau)+1+w_v(\tau)-1 +c_ v(\tau)
                    =\alpha_v(\tau)$$
                \item
                    That $update$ does not contain $v$, then $w_v(\tau')= w_v(\tau)$, 
                    the update may erase at most one element of $R$ then
                    $c_v(\tau')\geq c_v(\tau)-1$ and by Observation~\ref{obsA:updatescan},
                    $r_v(\tau')=r_v(\tau)+1$.  Hence:
                    $$\alpha_v(\tau')\geq r_v(\tau)+1+w_v(\tau) +c_ v(\tau)-1
                    =\alpha_v(\tau)$$
            \end{itemize}

            Consider the second case and the two following subcases:
            \begin{itemize}
                \item
                    That $update$ is an $update(-,V)$,  with $v\in V$ or $v \subseteq V$, 
                    then $w_v(\tau')= w_v(\tau)$.
                    As the update write $v$ in $R$, $c_v(\tau')\geq c_v(\tau)$.
                \item
                    That $update$ does not contain $v$, then $w_v(\tau')= w_v(\tau) +1$, 
                    the update may erase at most one element of $R$ then $c_v(\tau')\geq c_v(\tau)-1$.  
                    Hence: $$\alpha_v(\tau')\geq r_v(\tau)+w_v(\tau) +1+c_ v(\tau)-1 =\alpha_v(\tau)$$
            \end{itemize}
    \end{itemize}

\end{proof}
By an easy induction on the steps of $H$ we get:
\begin{lemma}(Lemma \ref{lemma:increasing0})\label{lemmaA:increasing0}
    If $\alpha_v(\tau) > n$  then for all $\tau'$, such that
    $\tau \leq \tau' $, $\alpha_v(\tau) \leq \alpha_v(\tau')$
\end{lemma}

\begin{lemma}(Lemma \ref{lemma:increasing})\label{lemmaA:increasing}
    If $\tau_v\neq \bot$ then   for all $\tau$,  such that $\tau_v\leq \tau $, 
    %$\alpha_v(\tau)>n$ and 
    $n<\alpha_v(\tau_v)\leq \alpha_v(\tau)$
\end{lemma}

\begin{proof}
    At time $ \tau_{v} $, $v$ is in all registers then  $c_v(\tau_v)=n$
    and $\alpha_v(\tau)\geq n$.
    Consider the update made just before time $\tau_v$, and the process that made
    this update, the next snapshot operation for this process is a $scan$
    or the process terminates with $v$ in its view, we have $r_v(\tau)+w_r(\tau) \ge 1$
    and then just after time $\tau_v$,  $\alpha_v(\tau)> n$. Hence by 
    Lemma~\ref{lemmaA:increasing0}, and an easy induction we deduce the Lemma.
\end{proof}

\begin{lemma}%(Lemma \ref{lemma:atleastone})
    \label{lemmaA:atleastone}
    If $\tau_{v} \neq \bot$, for all  $\tau \geq \tau_{v} $,  $c_v(\tau)\geq 1$ 
\end{lemma}
\begin{proof}
    From Lemma~\ref{lemmaA:increasing},
    $\alpha_v(\tau)=r_v(\tau)+w_v(\tau)+c_v(\tau) >n$ and by
    observation~\ref{obsA:trivial}  that $r_v(\tau)+w_v(\tau) \leq
    n$. Hence $c_v(\tau)\geq 1$.
\end{proof}

\begin{lemma}( Lemma \ref{lemma:safe} )\label{lemmaA:safe}
    $H_{seq}$ satisfies the sequential specification of the weak set
\end{lemma}
\begin{proof}
    Let $op$ be  a $\textsc{get}$ operation, consider the set $A$ of all $\textsc{add}$ operations
    linearized before (such that  $\textsc{add}(x) \in A$ if 
    $\tau_{\textsc{add}(x)}\leq \tau_{op}$). If $\textsc{get}$ does not terminate,
    there is no linearization. If $\textsc{get}$ operation $op$  terminates it
    returns with, say, 
    % We have to prove that the view
    % returned by $op$ is the set of all the values $x$ for  $add(x)\in A$.
    view $V$.
    % let $x\in V$.

    Consider $x\in V$. By observation~\ref{obsA:integrity}, then at least one process
    invoked an $\textsc{add}(x)$.

    Due to the termination loop condition of Line~\ref{li:while-get} of Algorithm of
    Fig.~\ref{fig:aws}, at the time $\tau$ of the last $scan$ made by any process 
    returning $V$, if $x\in V$, then $x$ belongs to all registers proving that 
    $\tau_x < \tau$.
    % hence there is an operation $add(x)$ in  $A$ linearized before $\tau_{op}$. 
    Then if $x\in V$, then there exists $\textsc{add}(x)$ in $A$.
    % $A$.

    Assume that there is a $v$ such that $\tau_v<\tau_{op}$, then by
    Lemma~\ref{lemmaA:atleastone}, after time $\tau_v$, $v$ belongs forever
    to at least one register in $R$. Then all $scan$ made after time $\tau_v$ contains $v$. 
    Due to the termination loop condition Line~\ref{li:while-get}, if any $scan$ of 
    a $\textsc{get}$ contains $v$, $v$ is in the view returned by that $\textsc{get}$. 
    Then if $op$ performs a $scan$ after time $\tau_v$, the view returned contains $v$.
    Then if $\textsc{add}(v)\in A$ then $v\in V$.

    Hence we have $V=\{x | add(x)\in A\}$.
\end{proof}

\subsection{$b$-Iterated barycentric agreement protocol}

To make the present paper self-contained, we present, in Fig.~\ref{alg:barycentric_agreement},
the $b$-iterated barycentric agreement protocol, which is a verbatim copy of one appeared 
in~\cite{yanagisawa2017wait}.

\begin{figure}[h]
    \hrule \vspace{1mm} {%\small
        \setcounter{linenumber}{0}
        \begin{tabbing}
            bbb\=bbb\=bbb\=bbb\=bbb\=bbb\=bbb\=bbb\=bbb\=bbb\=bbb\=bbb\=  \kill

            Shared variable : \\
            \>\texttt{array of atomic weak set objects} : $\weakset[0\ldots b-1]$ \\[2mm]

            \textsc{Code for a process} \\[2mm] 
            Local variable:\\
            \>\texttt{Integer} $i$ init 0\\
            \>\texttt{Value} $view$ init $\perp$\\[2mm]

            \textsc{BARYAGREE}$_{b}$($v$):\\
            \nnll\> $\mathit{view}=v$\\
            \nnll\> \textbf{for} $i=0,\ldots,b-1$ \textbf{do}\\
            \nnll\>\>   \weakset[$i$].\add{\mathit{view}}\\
            \nnll\>\>   $\mathit{view}=\weakset[i].\get$\\
            \nnll\> \textbf{return} $\mathit{view}$\\
        \end{tabbing}
        \vspace{-6mm}
    \hrule}
    \caption{Anonymous $b$-iterated barycentric agreement protocol}
    \label{alg:barycentric_agreement}
\end{figure}

In the protocol, each process starts with its private input value and assigns it 
to the local variable view (line 1). Thereafter, the process iterates from $i=0$ to $b$, 
the operation of adding its view to \weakset$[i]$ and updating its view by the
result of a get operation to \weakset[$i$] (line 2--4). 
At last, the process outputs the value held in its local variable $view$ (Line~5).

%%%%%%
%\renewcommand{\theheorem}{\roman{theorem}}
%\begin{lemma}%\label{lemma:terminaison}
%If the operation $\textsc{add}(v)$ terminates then before the end of this operation $v$ belongs to all registers of $R$.
%If the operation $\textsc{get}()$ terminates  and returns $V$ then before the end of this operation $V$ belongs to all registers of $R$.
%\end{lemma}
%\begin{proof}
%If the $ \textsc{add}(v)$ terminates then the condition Line~\ref{li:while-add} is false. 
%$v$ belongs to all cells of $Snap$. $Snap$ comes from the snapshot Line~\ref{li:scan1} or  Line~\ref{li:scan2}.
%When the process executes this line, $v$ belongs to all registers of $R$.
%
%If the $ \textsc{get}()$ terminates and returns $V$ then the condition Line~\ref{li:while-get} is false. 
%$V$ belongs to all cells of $Snap$. $Snap$ comes from the snapshot Line~\ref{li:scan3} or  Line~\ref{li:scan4}.
%When the process executes this line, $V$ belongs to all registers of $R$.
%\end{proof}

%%% Local Variables:
%%% mode: latex
%%% TeX-master: "paper"
%%% End:

\end{document}